\documentclass[11pt]{article}
\usepackage[T1]{fontenc}
\usepackage{lmodern}
\usepackage[a4paper,margin=22mm]{geometry}
\usepackage{amsmath,amssymb,amsthm}
\usepackage{booktabs,microtype}
\usepackage[hidelinks]{hyperref}
\usepackage{natbib}

\newcommand{\var}{\mathit{var}}
\newcommand{\cl}{\mathit{cl}}
\newcommand{\grd}{\mathit{grd}}
\newcommand{\fin}{\mathit{fin}}
\newcommand{\PSPACE}{\mathsf{PSPACE}}
\newcommand{\C}{\mathcal C}
\newcommand{\ind}{\mathbf 1}
\newtheorem{theorem}{Theorem}
\newtheorem{proposition}{Proposition}
\newtheorem{lemma}{Lemma}
\newtheorem{claim}{Claim}
\newtheorem{remark}{Remark}
\hypersetup{
  pdftitle={Subgame-Perfect Nash Equilibria of Plurality Voting with Abstention: a PSPACE hardness result for restricted ballots},
  pdfsubject={PSPACE-completeness of SPNE of Plurality Voting with Abstentions}
}
\title{Subgame-Perfect Nash Equilibria of Plurality Voting with Abstention: \\
a PSPACE-Completeness Result for Restricted Ballots}
\author{Edith Elkind\footnote{The construction in the proof was proposed by Astra. The exposition is due to EE.}\\
Northwestern University}
\date{September 2026}
\begin{document}
\maketitle
\begin{abstract}
We consider sequential Plurality elections in which each voter may abstain
or vote for a single candidate.
Each voter assigns utilities to all candidates; for each voter, this induces a (weak) order over the candidates.
Ties are resolved uniformly at random, and voting has a small positive cost, so that a voter prefers to abstain when their vote cannot change the election outcome. 
We consider a variant of this model where, for each voter, we additionally specify a prefix of her ranking, so that she is only allowed to vote for a candidate from that prefix (or abstain).
We prove that for this variant of the model, deciding whether a designated candidate is among the election winners
in a subgame-perfect equilibrium of the associated extensive-form game is $\PSPACE$-complete. This partially resolves an open problem from the work of \citet{DE10}.
\end{abstract}
\section{Background}
We start by presenting our model; our exposition follows that of \citet{DE10}.
\subsection{Basic Setup}
We consider elections with a set of candidates $\C$, $|\C|=m$, and a set of voters $N=\{v_1, \dots, v_n\}$, where each voter $v\in N$ is described 
by a utility function $U_v:\C\to\mathbb Q$, given explicitly in
binary. 

The voters submit their ballots sequentially, one by one; each voter observes all already-submitted ballots and then either votes for a single candidate in $\C$ or abstains. That is, for each voter the set of possible ballots is $\C\cup\{\bot\}$, where $\bot$ indicates abstention.
Let $b_v$ denote the ballot of voter $v\in N$. Then, for each candidate $c\in\C$, her Plurality score is computed as $s(c) = |\{v\in N: b_v=c\}|$. All candidates with the maximum Plurality score are declared to be the election winners (if all voters abstain, each candidate's Plurality score is 0, so all candidates in $\C$ are election winners); we denote the set of all election winners by $W$.

Voting carries a cost: each voter incurs a cost of $0$ if she abstains and a positive rational cost of $\varepsilon$ otherwise. Moreover, each voter assumes that the eventual winner is chosen from the winner set $W$ uniformly at random. 
Accordingly, we write
$$
U_v(W)=\frac{1}{|W|}\cdot\sum_{c\in W}U_v(c); \qquad
U^*_v(W) = U_v(W)-\varepsilon\cdot\ind_{\{b_v\ne\bot\}}; 
$$
we refer to $U_v(W)$ as the {\em value} that $v$ assigns to $W$ and to 
$U^*_v(W)$ as the {\em utility} that $v$ obtains from $W$.
Following \citet{DE10},  we assume that the utility function of each voter $v$ satisfies 
\begin{equation}\label{eq:distinct}
\frac{1}{|W|}\sum_{c\in W}U_v(c)\neq 
\frac{1}{|W'|}\sum_{c\in W'}U_v(c)
\end{equation}
for all non-empty distinct $W, W'\subseteq \C$, and, moreover, 
\begin{equation}\label{eq:cost}
  0<\varepsilon<
  \min_{\substack{v\in N,\ \varnothing\ne W,W'\subseteq\C\\W\ne W'}}
  \left|\frac{1}{|W|}\sum_{c\in W}U_v(c) -  
\frac{1}{|W'|}\sum_{c\in W'}U_v(c)\right|.
\end{equation}
In particular, each voter prefers to abstain if their vote does not change the election outcome.

\subsection{Games and Equilibria}
Assume that the voter order is $v_1, \dots, v_n$. Together with the voters' utility functions, this order induces an extensive-form game $G$ among the voters, where
a {\em pure strategy} $\gamma_i$ for a voter $v_i$ specifies a ballot $b_{v_i}\in\C\cup\{\bot\}$ given each possible {\em history} $b_{v_1}, \dots, b_{v_{i-1}}$.
A strategy profile $(\gamma_1, \dots, \gamma_n)$ induces a {\em ballot profile}
$(b_{v_1}, \dots, b_{v_n})$: for each voter $i$, her ballot $b_{v_i}$ follows
the prescription of $\gamma_i$ at history $b_{v_1}, \dots, b_{v_{i-1}}$.
In turn, a ballot profile induces a winner set $W$: $W$ consists of candidates with the highest Plurality score according to $(b_{v_1}, \dots, b_{v_n})$. Given $W$, we can compute 
the utility $U^*_v(W)$ of each voter $v\in N$.
Thus, each strategy profile is associated with a utility for each voter. 

A strategy profile is a {\em Nash equilibrium} if no voter can increase her utility by unilaterally changing her strategy. While Nash equilibrium is a standard solution concept in normal-form games, in extensive-form games one usually considers a refinement of Nash equilibrium, known as {\em subgame-perfect Nash equilibrium}. To introduce this concept in our setting, we need additional definitions.

Fix a voting game $G$, as defined in this section. For each $i\in [n]$, 
each history $b_{v_1}, \dots, b_{v_{i-1}}$ defines a {\em continuation game}
$G(b_{v_1}, \dots, b_{v_{i-1}})$; this is an extensive-form game where the set of voters
is $i, \dots, n$, but the Plurality score of each candidate is computed using both
the initial $i-1$ ballots and the $n-i+1$ ballots cast by the voters. A pure strategy
$\gamma_i$ of a voter $i$ in $G$ can be restricted to a continuation game $G(b_{v_1}, \dots, b_{v_{i'}})$ for any $i'<i$ in a natural way: for each history $b_{v_{i'+1}}, \dots, b_{v_{i-1}}$ in the continuation game, the restricted strategy of $i$ proposes 
the same action as proposed by $\gamma_i$ in $G$ for history $b_{v_1}, \dots, b_{v_{i-1}}$. 
A strategy profile $(\gamma_1, \dots, \gamma_n)$ in $G$ is a {\em subgame-perfect Nash equilibrium (SPNE)} if for every $i\in [n]$ and every history $b_{v_1}, \dots, b_{v_{i-1}}$, 
the restrictions of $\gamma_i, \dots, \gamma_n$ to the continuation game
$G(b_{v_1}, \dots, b_{v_{i-1}})$ form a Nash equilibrium of that game.

Plurality voting games considered in this paper always admit subgame-perfect Nash equilibria (SPNE); such equilibria can be computed (albeit inefficiently) by backward induction. Note that a strategy in a voting game is an exponential-sized object, as it needs to prescribe a ballot for each history. Thus, computing a SPNE is not a sensible computational problem.
However, each SPNE corresponds to a set of winners, which is a `small' object, and a natural computational problem is to compute this set.

\citet{DE10} conjecture that the following problem is $\PSPACE$-complete: given an election with a set of candidates $\C$, a set of voters $N$, a list $(u_v)_{v\in N}$ of voters' utility functions that satisfy assumptions~\eqref{eq:distinct} and~\eqref{eq:cost}, and a voting cost $\varepsilon$, determine whether a designated candidate $d$ belongs to the winner set in an SPNE of the associated extensive-form game. We note that 
under assumptions \eqref{eq:distinct} and \eqref{eq:cost} a voting game may have multiple SPNE, but they all have the same winner set (as no voter is indifferent between two winner sets). 

We consider a variant of this problem where we additionally specify, for each $v\in N$, a positive integer $k_v\in [m]$, so that each voter $v$ is only allowed to vote for one of her top $k_v$ candidates (or abstain), where it is assumed that a voter ranks all candidates by utility.
We refer to this problem as {\sc Plurality SPNE Winner with Prefix Ballots}.
%%%%%%%%%%%%%%%%%%%%%%%%%%%%%%%%%%

\section{Main Result}

\begin{theorem}\label{thm:main}
\textnormal{\textsc{Plurality SPNE Winner with Prefix Ballots}} is \(\PSPACE\)-complete.
\end{theorem}
We prove the theorem in two parts. First, we establish that our problem is in $\PSPACE$. Then we will show that it is $\PSPACE$-hard.

\begin{proposition}\label{prop:inpspace}
\textnormal{\textsc{Plurality SPNE Winner with Prefix Ballots}} is in $\PSPACE$.
\end{proposition}
\begin{proof}
The standard algorithm for computing a subgame-perfect Nash equilibrium in an extensive-form
game is backward induction. In our setting, backward induction operates as follows.
First, for each possible history $b_{v_1}, \dots, b_{v_{n-1}}$, we compute an optimal action of voter $v_n$. Then, inductively, having computed the strategy for voter $v_{i+1}$ for all possible histories, we can compute the strategy of voter $v_i$ for all possible histories:
for a given history $b_{v_1}, \dots, b_{v_{i-1}}$, we consider all possible actions of $v_i$, 
and for each action we compute the eventual winner set given the already-computed strategies of voters $v_{i+1}, \dots, v_n$; we then pick an action that results in the highest utility, breaking ties arbitrarily (note that tie-breaking is not important, since all SPNE have the same winner set). That is, we implicitly construct the {\em game tree}, where each history is associated with a node, and the parent of a node $b_{v_1}, \dots, b_{v_i}$ is $b_{v_1}, \dots, b_{v_{i-1}}$, and traverse this tree from the leaves towards the root, computing an optimal action (and the associated winner set) at each node, and outputting the winner set at the root.

Clearly, a straightforward implementation of this algorithm requires exponential space.
However, we can explore the game tree using depth-first search, generating successor nodes on the fly.
Since the length of each history is at most $n$, and the size of the winner set is at most $m$, the algorithm uses polynomial space and correctly decides if $d$ is among the winners in some SPNE.
\end{proof}

For the hardness result, we reduce from the classic TQBF problem. An instance of this problem is given by a quantified Boolean formula
\begin{equation}\label{eq:qbf}
  \Phi=Q_1x_1\cdots Q_\nu x_\nu\,\bigwedge_{j=1}^{\mu}C_j,
\end{equation}
where $Q_i\in\{\exists,\forall\}$ for all $i\in [\nu]$, 
and each clause $C_j$, $j\in [\mu]$, is a disjunction of literals from $\{x_i, \overline{x}_i: i\in [\nu]\}$;
it is a yes-instance if $\Phi$ is true and a no-instance otherwise. This problem is $\PSPACE$-complete~\citep{SM73}.
We will additionally assume that our instance has the following properties: $\mu\ge 3$,  each clause is
nonempty, has at most three literals, and contains no variable twice; the problem is known to remain $\PSPACE$-hard under these constraints.

\subsection{Candidates, Voters, and Utilities}

Given a quantified Boolean formula $\Phi = Q_1x_1\cdots Q_\nu x_\nu\, \Psi$ with $\Psi = C_1\wedge\dots\wedge C_\mu$, we set $n=\nu +2\mu +2$, and construct an election
with $n$ voters and $m=n$ candidates.

The candidate set $\C$ consists of four components: we have
$\C=A\cup\mathit{CL}\cup G\cup\{f\}$, 
where $A = \{a_i:i\in[\nu]\}$ is the set of {\em variable candidates}, 
$\mathit{CL}=\{c_j:j\in[\mu]\}$ is the set of {\em clause candidates}, and  
$G=\{g_j:j\in[\mu+1]\}$ is the set of {\em guard candidates}; we refer to $f$ as the {\em final candidate}.
Note that $|\C|=n$. 

There are also $n$ voters: 
$\nu$ {\em variable voters} $\var_1,\ldots,\var_\nu$ (where voter $\var_i$ corresponds to the variable $x_i$), 
$\mu$ {\em clause voters} $\cl_1, \dots, \cl_\mu$ (where voter $\cl_j$ corresponds to the clause $C_j$), 
$\mu+1$ {\em guard voters} $\grd_1,\ldots,\grd_{\mu+1}$,  
and a single {\em final voter} $\fin$. They vote in this order.

The sets of permitted candidates for each voter are given in
Table~\ref{tab:actions}; later, we will specify the voters' utilities and the reader will be able to verify that each voter weakly prefers her permitted  candidates over all other candidates. Recall that each voter can also abstain.

\begin{table}[ht]
\centering
\begin{tabular}{@{}ll@{}}
\toprule
Voter & Permitted candidates\\
\midrule
$\var_i$ & $\{a_i\}$\\
$\cl_j$  & $\{c_j,f\}\cup\{a_i:x_i\in C_j\}$\\
$\grd_j$ & $\{g_j,f\}$\\
$\fin$ & $\{f\}$\\
\bottomrule
\end{tabular}
\caption{Permitted ballots.}
\label{tab:actions}
\end{table}

Before we specify the voters' utilities, we outline the key idea of the proof.

For each ballot profile of variable and clause voters, we consider the continuation
game among the guard voters and the final voter, and argue that in its SPNE
the winner set is either $\{f\}$ or a set $W$ that contains 
$f$ and all guard candidates, with all winners receiving exactly one vote; in particular, if any of the clause voters votes 
for $f$, then the outcome is $\{f\}$. 
We will then prove that $\Phi$ is false if and only if the winner set in all SPNE is $\{f\}$.

Intuitively, a variable voter $\var_i$ decides whether $x_i$ should be set to true:
voting for $a_i$ corresponds to setting $x_i$ to true and abstaining corresponds to setting $x_i$ to false. The variable voters' utilities are such that $\exists$-voters want to make $\Psi$ true, 
while $\forall$-voters want to make $\Psi$ false.
To implement this, an $\exists$-voter $\var_i$ assigns 
positive utility to all candidates other than $f$ and utility 0 to $f$ (so that $\{f\}$ is her least preferred outcome), whereas a $\forall$-voter $\var_i$ assigns similar utilities to $a_i$ and $f$, with $a_i$ being slightly more attractive, and utility 0
to other candidates (so that she prefers $\{f\}$ to a tie that involves $a_i$ and many low-utility candidates, where the utility of $a_i$ is `diluted').
Now, any ballot profile for the variable voters can be interpreted as a truth assignment:
$x_i$ is set to true if $\var_i$ votes for $a_i$ and $x_i$ is set to false if $\var_i$
abstains.

We now want a clause voter $\cl_i$ to have the following preferences: if her clause is not satisfied by the truth assignment defined by the variable voters, she prefers $f$ to the outcome that includes all candidates except for the variable candidates with zero votes (a `large tie'), and if it is satisfied, she prefers the `large tie' to $f$. In equilibrium, this results in all clause voters voting for their clause candidate if the truth assignment defined by the variable voters makes $\Psi$ true, but acting in a way that gets $f$ elected as a unique winner otherwise (e.g., some clause voter may vote for $f$ or for a variable candidate with a positive score, which then results in guard voters enforcing $\{f\}$ as the election outcome).
 We achieve this by making $\cl_i$ assign positive utility to $f$ and all candidates in $A$, with candidates associated with positive literals in clause $C_i$ having slightly higher utility; this way, if a clause is satisfied, the presence 
of positive literal candidates/absence of negative literal candidates makes a large tie more attractive than $f$, 
but if it is not satisfied, the average utility of the tie is lower, 
so $f$ becomes more attractive.

We now specify the voters' utilities so as to implement this idea. 
Let $i\in [\nu]$.
If $Q_i=\exists$, let
\[
  U_{\var_i}(a_i)=2,\qquad U_{\var_i}(f)=0,\qquad
  U_{\var_i}(c)=1 \text{ for }c\notin\{a_i,f\}.
\]
If $Q_i=\forall$, let
\[
  U_{\var_i}(a_i)=2,\qquad U_{\var_i}(f)=1,\qquad
  U_{\var_i}(c)=0 \text{ for }c\notin\{a_i,f\}.
\]
For a clause $C_j$, $j\in [\mu]$, 
let $\mathrm{neg}_j$ be the number of negative literals in $C_j$, and set
\[
  \gamma_{ij}=
  \begin{cases}
    1,&x_i\in C_j,\\
    -1,&\neg x_i\in C_j,\\
    0,&\text{otherwise}.
  \end{cases}
\]
The utility of voter $\cl_j$ is then given by
\begin{equation}\label{eq:clause-utilities}
\begin{aligned}
  U_{\cl_j}(a_i)&=\nu+1+\gamma_{ij},&
  U_{\cl_j}(c_j)&=(2\mu+1)(\nu+1)+\mathrm{neg}_j-\tfrac12,&
  U_{\cl_j}(f)&=\nu+1,
\end{aligned}
\end{equation}
and 0 for every other clause candidate and every $g\in G$.
For a guard voter $\grd_j$ with $j\in[\mu+1]$, let
\[
  U_{\grd_j}(g_j)=2n,\qquad U_{\grd_j}(f)=1,\qquad
  U_{\grd_j}(c)=0\text{ for }c\notin\{g_j,f\}.
\]
Finally, let $U_{\fin}(f)=1$ and $U_{\fin}(c)=0$ for $c\ne f$.

We note that these utilities do not satisfy assumption
\eqref{eq:distinct}. We will explain how to fix this issue towards the end of the proof.

Let $\varepsilon=\frac{1}{n^2}$ and let 
the designated candidate be $g_1$.
To complete the proof, 
we will argue that if $\Phi$ is false then in any SPNE of our game the winner set 
is $\{f\}$, and if it is true then in any SPNE of our game the winner set contains all candidates in $G$. We split the proof into a sequence of lemmas.

\subsection{Key Lemmas}

Our first lemma characterizes the equilibrium behavior of the guard voters.

\begin{lemma}\label{lem:guard}
Consider any sequence of ballots $b_{\var_1}, \dots, b_{\var_\nu}, b_{\cl_1}, \dots, b_{\cl_\mu}$, and for each candidate $c\in\C$ let $s'(c)$ denote the Plurality score 
that $c$ obtains from these ballots. Then in any SPNE of 
$G(b_{\var_1}, \dots, b_{\var_\nu}, b_{\cl_1}, \dots, b_{\cl_\mu})$
the winner set $W$ is of the following form:
\[
  \begin{cases}
    \{f\},& \text{if }s'(f)\geq1\text{ or } s'(c)\geq2\text{ for some $c\neq f$},\\
    \{c:s'(c)=1\}\cup G \cup \{f\}, & \text{otherwise}.
  \end{cases}
\]
In the second case, each guard voter $\grd_j$ votes for $g_j$ and
the final voter $\fin$ votes for $f$.
\end{lemma}
\begin{proof}
 Observe that each variable candidate receives at most one vote from the respective variable voter and at most $\mu$ votes from clause voters, whereas each clause candidate receives at most one vote from the respective clause voter; the final candidate $f$ receives at most $\mu$ votes from the clause voters.
 Hence, we have $s'(a)\leq \mu+1$ for all $a\in A$, $s'(c)\leq 1$ for all $c\in\mathit{CL}$,
and $s'(g)=0$ for all $g\in G$. In particular,
no candidate gets more than $\mu+1$ votes from the first $\nu+\mu$ ballots.

Suppose first that some candidate in $A$ gets at least two votes from the first $\nu+\mu$ ballots, and let $s = \max\{s'(c): c\in A\}$. In this case, no $g\in G$ can ever be a winner,
since it can receive at most one vote. Moreover, no candidate in $A\cup\mathit{CL}$ can get additional votes. Finally, all guard voters as well as the final voter prefer $f$ to all candidates in $A\cup\mathit{CL}$, and if they all vote for $f$, then $f$ gets $\mu+2>s$ points. Consequently, in equilibrium
the first $\mu+2-\max\{s+1-s'(f), 0\}$ of these voters will abstain, and the remaining  
voters (if any) will all vote for $f$, so that $f$ gets $\max\{s'(f), s+1\}$ votes and becomes the unique winner.  

Now, suppose that $s'(c)\le 1$ for each $c\in A\cup\mathit{CL}$. If $s'(f)\ge 2$, all remaining voters will abstain, and $f$ will become the unique winner. If $s'(f)=1$, all guard voters will abstain, and $\fin$ will abstain (if no other candidate has positive score) or vote for $f$; again, $f$ becomes the unique winner.

Finally, suppose that $s'(f)=0$.
 Note that $\fin$ votes for $f$ if $f$ receives no votes from the guard voters (as this places $f$ in the winner set, or makes it the unique winner if all previous voters abstain) or if $f$ receives exactly one vote from the guard voters and at least one other candidate has a positive Plurality score (as this makes $f$ the unique winner instead of being tied with other candidates).
 If $f$ receives two or more votes from the guard voters or if $f$ is the only candidate with a positive score, $\fin$
 abstains: no candidate other than $f$ gets more than one vote, so $f$ is the unique winner anyway, and $\fin$ prefers not to incur the cost of voting.
We now claim that each guard voter $\grd_j$ abstains 
if any of the earlier guard voters votes for $f$ and votes for $g_j$ otherwise. Indeed, it is immediate that $\grd_j$ abstains if $f$ receives a vote before $\grd_j$'s turn, as $\fin$ would then ensure that $f$ is the unique winner anyway, so voting for $g_j$ wastes the voting cost $\varepsilon$.
The second part of our claim follows by backward induction.
For the base case, note that $\grd_{\mu+1}$ votes for $g_{\mu+1}$ if $f$ receives no votes, as this guarantees that $g_{\mu+1}$ will be in the winner set (this gives $\grd_{\mu+1}$ a utility of at least $2-\varepsilon$, which is higher than her utility from any winner set that does not contain $g_{\mu+1}$), whereas abstaining or voting for $f$ does not.
Inductively, consider a voter $\grd_j$, $j\in [\mu]$. If $f$ receives no votes from the earlier guard voters, $\grd_j$ should vote for $g_j$: by the inductive hypothesis, all subsequent guard voters will vote for their top candidates and $g_j$ will be in the winner set---an outcome $\grd_j$ prefers to any outcome where $g_j$ is not in the winning set.
\end{proof}

Our second lemma argues that, in an SPNE, a clause voter either votes so that
the eventual outcome is $f$ or she votes for $c_j$.

\begin{lemma}\label{lem:clause}
Let $j\in [\mu]$, and
consider a sequence of ballots ${\mathbf b} = b_{\var_1}, \dots, b_{\var_\nu}, b_{\cl_1}, \dots, b_{\cl_{j}}$.
Let $W$ be the winner set in an SPNE of the continuation game $G(\mathbf{b})$. Then
\begin{itemize}
\item[(1)] if $b_{\cl_j} = f$ then $W = \{f\}$ and  $U^*_{\cl_j}(W) = \nu+1-\varepsilon$. 
\item[(2)]
If $b_{\cl_j} =c_j$ then either 
$W = \{f\}$ or $G\cup\{c_j, f\}\subseteq W$.
\item[(3)]
If $b_{\cl_j} \neq f, c_j$
then either $W=\{f\}$ or $U^*_{\cl_j}(W) < \nu+1-\varepsilon$.
\end{itemize}
\end{lemma}
\begin{proof}
The first claim follows directly from Lemma~\ref{lem:guard}.
For the second claim, suppose that $W\neq\{f\}$. Then by Lemma~\ref{lem:guard} we have $G\cup\{f\}\subseteq W$, and, moreover,
$W$ contains every candidate that receives at least one Plurality vote;
this includes $c_j$. 
Finally, suppose that $\cl_j$ votes for some $a_i$ with $x_i\in C_j$ or abstains. Then, by Lemma~\ref{lem:guard}, 
the outcome is either $\{f\}$ 
or a winner set $W$ with $c_j\not\in W$ (as no voter other than $\cl_j$ can vote for $c_j$) and $G\cup\{f\}\subseteq W$. 
In the latter case, let $p_j$ denote the number of candidates $a_i$ in $A\cap W$ such that $x_i\in C_j$; note that $p_j\le 3$. Then $|W|\ge |A\cap W|+|G|+1$ and 
\begin{align*}
U^*_{\cl_j}(W) - (\nu+1) +\varepsilon &\le\frac{p_j\cdot (\nu+2)+(|A\cap W|-p_j)\cdot(\nu+1)+(\nu+1)}{|A\cap W|+|G|+1} - (\nu+1) + \varepsilon\\
&= 
\frac{p_j+(|A\cap W|+1-|A\cap W|-\mu-2)\cdot(\nu+1)}{|A\cap W|+\mu+2}  +\varepsilon < 0, 
\end{align*}
where the final transition follows from $\mu\ge 3$ and our choice of $\varepsilon$.
\end{proof}

Lemma~\ref{lem:clause} does not fully describe the behavior 
of the clause voters; in particular, it does not explain how a clause voter's
behavior depends on whether her clause is true or false under the assignment
defined by the choices of the variable voters. We will now focus on this aspect
of the clause voters' choices.

Fix an assignment $\sigma\in\{0,1\}^\nu$, and write
\[
  W_\sigma=\{a_i:\sigma_i =  1\}\cup\mathit{CL}\cup G \cup\{f\}, \qquad w = |\{a_i:\sigma_i =  1\}|. 
\]
Given a quantifier-free formula $\Xi$ over variables $x_1, \dots, x_\nu$, we write $\sigma\models\Xi$ if 
$\sigma$ satisfies $\Xi$ and $\sigma\not\models\Xi$ otherwise.

\begin{lemma}\label{lem:sigmavsf}
Consider an assignment $\sigma\in\{0,1\}^\nu$ and a $j\in [\mu]$. If $\sigma\models C_j$ then $U_{\cl_j}(W_\sigma) - U_{\cl_j}(\{f\})>\epsilon$ 
and if $\sigma\not\models C_j$
then $U_{\cl_j}(\{f\}) - U_{\cl_j}(W_\sigma) >\epsilon$.
\end{lemma}
\begin{proof}
Fix an assignment $\sigma$ and a $j\in [\mu]$.
let $z_j(\sigma)$ be the number of satisfied literals in $C_j$ under $\sigma$.
We have 
\begin{align*}
z_j(\sigma) & =|\{a_i\in A: \overline{x}_i\in C_j, \sigma_i=0\}| + |\{a_i\in A: x_i\in C_j, \sigma_i=1\}| \\
&= 
|\{a_i\in A: \overline{x}_i\in C_j\}| - |\{a_i\in A: \overline{x}_i\in C_j, \sigma_i=1\}| + |\{a_i\in A: x_i\in C_j, \sigma_i=1\}|\\
&= \mathrm{neg}_j+\sum_{i:\sigma_i=1}\gamma_{ij}, 
\end{align*}
and
\begin{align*}
  U_{\cl_j}(W_\sigma)
  &=\frac{w\cdot (\nu+1)+\sum_{i:\sigma_i=1}\gamma_{ij}+(2\mu+1)(\nu+1)+\mathrm{neg}_j-\tfrac12+(\nu+1)}
         {w+2\mu+2}\\
  &=(\nu+1)+\frac{1}{w+2\mu+2}\cdot(z_j(\sigma)-\tfrac12).
\end{align*}
Recall that $U_{\cl_j}(\{f\})=\nu+1$.
Moreover, $z_j(\sigma)-\frac12\ge \tfrac12$  if $\sigma\models C_j$ and $z_j(\sigma)-\frac12 = -\tfrac12$  if $\sigma\not\models C_j$.
Hence $|U_{\cl_j}(W_\sigma)-U_{\cl_j}(\{f\})| > \varepsilon$, and 
$U_{\cl_j}(W_\sigma)-U_{\cl_j}(\{f\})>0$
if and only if $\sigma\models C_j$. 
\end{proof}

We are now ready to describe the collective behavior of the clause voters.

\begin{lemma}\label{lem:formula}
Consider a sequence of ballots 
${\mathbf b} = b_{\var_1}, \dots, b_{\var_\nu}$, and
define an assignment $\sigma\in\{0, 1\}^\nu$ by setting $\sigma_i=1$ if $b_{\var_i}=a_i$ 
and $\sigma_i=0$ otherwise.
Let $W$ be the winner set in an SPNE of the continuation game
$G({\mathbf b})$.
Then $W=W_\sigma$ if $\sigma\models \Psi$ and $W=\{f\}$ if $\sigma\not\models\Psi$.
\end{lemma}
\begin{proof}
First, we will argue that either $W=\{f\}$ or $W=W_\sigma$. Indeed, suppose that $W\neq\{f\}$. By Lemma~\ref{lem:guard}, 
it follows that no clause voter votes for $f$. Moreover, if some clause voter votes for 
a candidate $a\in A$ or abstains, then, by Lemma~\ref{lem:clause}, she can benefit from voting for $f$ (which would make $f$ the unique winner). Hence, every clause voter votes for their respective clause candidate, and therefore $W=W_\sigma$. 

We consider two cases.
\begin{itemize}
\item[$\sigma\not\models\Psi$] Let $C_j$ be a clause that is not satisfied under $\sigma$. 
Suppose for the sake of contradiction that $W=W_\sigma$. We then have $c_j\in W$. This means that voter $\cl_j$ votes for $c_j$, as no other voter can vote for $c_j$. But since 
$C_j$ is not satisfied, 
it follows from Lemma~\ref{lem:sigmavsf} that voter $\cl_j$ would benefit from changing her vote to $f$, thereby making $f$ the unique winner. 
Hence, $W\neq W_\sigma$.

\item[$\sigma\models\Psi$] 
 We will argue that every clause voter votes for her clause candidate; by Lemma~\ref{lem:guard}, this implies $W=W_\sigma$.

We proceed by backward induction on $j$: for $j=\mu, \dots, 1$, 
we show that 
if the first $j-1$ clause voters vote for their respective clause candidates then in every 
SPNE of the resulting continuation game voter $\cl_j$ also votes for her clause candidate $c_j$.

For $j=\mu$,  $\cl_\mu$ can ensure that the winner set is $W_\sigma$
by voting for $c_\mu$. Since $\sigma$ satisfies $C_\mu$, Lemma~\ref{lem:sigmavsf} implies that this is more attractive than $\{f\}$, which, in turn (by Lemma~\ref{lem:clause}) is more attractive than any other outcome that can be achieved by voting for a candidate in $A$ or abstaining. 

Now, suppose we have established our claim for $j+1$. Then if the clause voter $\cl_j$ were to vote for $c_j$, we would be in the situation captured by the inductive hypothesis and hence all subsequent clause voters would vote for their clause candidates
and the outcome would be $W_\sigma$. 
Otherwise, the winner set would be either $\{f\}$ or a set $W'$ such that $\cl_j$ prefers $\{f\}$ to $W'$ (by Lemma~\ref{lem:clause}).
By Lemma~\ref{lem:sigmavsf},  $\cl_j$ prefers $W_\sigma$ to $\{f\}$ when $C_j$ is satisfied, so she chooses to vote for~$c_j$.

This completes the inductive step, so 
we have argued that in an SPNE 
of $G({\mathbf b})$ every clause voter votes for her clause candidate. By Lemma~\ref{lem:guard}, this implies $W=W_\sigma$.
\end{itemize}
\end{proof}

Finally, we analyze the behavior of the variable voters. For each $i\in [\nu]$ we say that $\var_i$ is an {\em $\exists$-voter} if $Q_i=\exists$ and a {\em $\forall$-voter} if $Q_i=\forall$.

\begin{lemma}\label{lem:var}
   For each assignment $\sigma\in\{0,1\}^\nu$ and each $i\in [\nu]$, if $\var_i$ is an $\exists$-voter then 
   $U_{\var_i}(W_\sigma)>U_{\var_i}(\{f\})+\varepsilon$ and if $\var_i$ is a $\forall$-voter then 
   $U_{\var_i}(\{f\}) > U_{\var_i}(W_\sigma)+\varepsilon$.
\end{lemma}
\begin{proof}
  Note that $\exists$-voters assign utility $1$ to each $g\in G$ and non-negative utilities to all other candidates, so for an $\exists$-voter $\var_i$ we have $U_{\var_i}(W_\sigma)\ge \frac{1}{n}$. On the other hand, we have $U_{\var_i}(\{f\})=0$, so our claim follows.

  In contrast, $\forall$-voters assign utility $1$ to $f$, utility $2$ to their top candidate, and utility $0$ to all other candidates. Since $G\subseteq W_\sigma$ and $|G|=\mu+1$,  for a $\forall$-voter $\var_i$ we have $U_{\var_i}(W_\sigma)\le \frac{1+2}{\mu+1}$. Since $\mu\ge 3$, our claim follows.
\end{proof}

\section{Putting It All Together}
We first need to introduce some notation for partial assignments, i.e., assignments to $i<\nu$ variables in 
$\{x_1, \dots, x_\nu\}$.

Given a partial assignment $\sigma\in\{0, 1\}^i$ for $i\in\{0, \dots, \nu\}$, 
let $\Phi[\sigma] = Q_{i+1} x_{i+1}\dots Q_\nu x_\nu\bigwedge_{j=1}^\mu C_j[\sigma]$, where for each $j\in [\mu]$ the clause $C_j[\sigma]$ is obtained by replacing each occurrence of $x_\ell$, $\ell\le i$, with $\sigma_\ell$ and each occurrence of 
$\overline{x}_\ell$, $\ell\le i$, with $1-\sigma_\ell$.

Given two partial assignments 
$\sigma\in\{0, 1\}^i$ and $\sigma'\in\{0, 1\}^{i'}$ for $0\le i<i'\le \nu$, we say that $\sigma'$ {\em extends} $\sigma$ if $\sigma_\ell =\sigma'_\ell$ for all $\ell\in [i]$. Further, for $i\in [\nu]$ and a history
${\mathbf b} = b_{\var_1}, \dots, b_{\var_i}$, we define a partial assignment $\sigma^{\mathbf b}\in\{0, 1\}^i$ by $\sigma^{\mathbf b}_\ell=1$ if $b_{\var_\ell}=a_\ell$
    and $\sigma^{\mathbf b}_\ell=0$ if $b_{\var_\ell}=\bot$; we extend this definition to empty histories by setting $\sigma^{()}$ to be the empty assignment.
    
We are now ready to complete the proof of the theorem.

\begin{proof}[Proof of Theorem 1]
We will now prove the following claim, by backward induction on $i$.

\begin{claim}\label{clm:induction}
    For each $i\in [\nu]$ and every history ${\mathbf b} = (b_{\var_1}, \dots, b_{\var_{i-1}})$, the winner set in all SPNE of the continuation game $G({\mathbf b})$ is $W_\sigma$
    for some $\sigma$ that extends $\sigma^{\mathbf b}$ if the formula 
    $\Phi[\sigma^{\mathbf b}]$ is true 
    and $\{f\}$ if $\Phi[\sigma^{\mathbf b}]$ is false.
\end{claim}
\begin{proof}
We first consider the base case $i=\nu$.

Suppose first that $Q_\nu=\exists$, i.e., $\var_\nu$ is an $\exists$-voter. If $\Phi[\sigma^{\mathbf b}]$ is true, then there is a way to set $x_\nu$
so that the resulting assignment satisfies all clauses.
Then $\var_{\nu}$ prefers to vote accordingly. Indeed, by Lemma~\ref{lem:formula} this would result in an outcome of the form $W_\sigma$, while voting in a way that does not satisfy some clause would result in $\{f\}$, and an $\exists$-voter prefers the former to the latter. 
If $\Phi[\sigma^{\mathbf b}]$ is false, 
then there is no way to set $x_\nu$ so as to satisfy all clauses, i.e., no matter how $\var_\nu$ votes, some clause is unsatisfied. By Lemma~\ref{lem:formula}, the resulting outcome is $\{f\}$.

Now, suppose that $Q_\nu=\forall$, i.e., $\var_\nu$ is a $\forall$-voter. If $\Phi[\sigma^{\mathbf b}]$ is true,
then, no matter how $\var_\nu$ votes, the assignment $\sigma$ that corresponds to the first $\nu$ votes satisfies all clauses and hence by Lemma~\ref{lem:formula} the resulting outcome is $W_\sigma$.
If $\Phi[\sigma^{\mathbf b}]$ is false, 
then for some value of $x_i$ some clause is not satisfied; $\var_\nu$ then prefers to vote accordingly, to ensure that the outcome is $\{f\}$. 

This completes the analysis for the base case. Now, suppose the claim has been proved for $i+1$; we will prove it for $i$.  Consider assignments $\sigma^0, \sigma^1\in \{0, 1\}^i$ that extend $\sigma^{\mathbf b}$
by setting the $i$-th bit to $0$ or $1$, 
respectively.
Again, the analysis depends on whether $Q_i=\exists$ or $Q_i=\forall$.

Suppose $Q_i=\exists$.
If $\Phi[\sigma^{\mathbf b}]$ is true, 
then at least one of the formulas
$\Phi[\sigma^0]$,  $\Phi[\sigma^1]$
is true. Hence, $\var_i$ can vote in such a way that for the resulting partial assignment $\sigma'\in\{\sigma^0, \sigma^1\}$ it holds that $\Phi[\sigma']$ is true.
By the inductive hypothesis, this results in a winner set $W_\sigma$ where $\sigma$ extends $\sigma'$. If $\var_i$ votes so that for the resulting partial assignment $\sigma'$ the formula $\Phi[\sigma']$ is false, 
by the inductive hypothesis, the eventual winner set is $\{f\}$. Since $\var_i$ prefers the former outcome to the latter,
she votes so that the winner set is some $W_\sigma$.

If $Q_i=\exists$ and $\Phi[\sigma^{\mathbf b}]$ is false, 
then, no matter how $\var_i$ votes, for the resulting partial assignment $\sigma'\in\{\sigma^0, \sigma^1\}$ the formula $\Phi[\sigma']$ is false. By the inductive hypothesis, this means that the outcome is~$\{f\}$.

Now, suppose that $Q_i=\forall$. If $\Phi[\sigma^{\mathbf b}]$ is true, 
then, no matter how $\var_i$ votes, for the resulting partial assignment $\sigma'\in\{\sigma^0, \sigma^1\}$ the formula $\Phi[\sigma']$ is true. By the inductive hypothesis, this means that the outcome is $W_\sigma$ for some $\sigma$ that extends $\sigma'$.

Finally, if $Q_i=\forall$ and $\Phi[\sigma^{\mathbf b}]$ is false,
then $\var_i$ can vote so that for the resulting partial assignment $\sigma'\in\{\sigma^0, \sigma^1\}$ the formula $\Phi[\sigma']$ is false. By the inductive hypothesis, this means that the outcome is~$\{f\}$.
\end{proof}
By taking $i=1$ in Claim~\ref{clm:induction}, we conclude that the winner set in all SPNE of our game is $W_\sigma$ for some $\sigma\in\{0, 1\}^\nu$ if $\Phi$ is true and $\{f\}$ if $\Phi$ is false.
This establishes the correctness of our reduction. 

It remains to explain how to modify our instance
to make it satisfy assumptions~\eqref{eq:distinct}
and \eqref{eq:cost}.
To this end, we pick a small enough 
rational value $\delta\in (0, 1)$ (e.g., $\delta=\frac{1}{8n^2}$ will work), make each voter rank the candidates in some way that 
is consistent with their utility function (so that $U_v(c)>U_v(c')$ implies that $v$ ranks $c$ above $c'$) and so that they rank their permitted candidates above the non-permitted candidates, and then,  
for each voter $v$, add $\delta^i$ to the utility that $v$ assigns to the candidate it ranks in the $i$-th position. We then set $\varepsilon = \frac{\delta^{n+1}}{n^2}$.  
It can be verified that for sufficiently small $\delta$ this modification does not change the argument, and the size of the modified instance is polynomial in the size of the original instance, and the complexity result holds under assumptions \eqref{eq:distinct} and \eqref{eq:cost}.
Together with Proposition~\ref{prop:inpspace}, this completes the proof.
\end{proof}
\begin{remark}
 In the voting game we constructed, all SPNE have the same winner set. This is because the set of possible outcomes consists of $\{f\}$ and sets of the form $W_\sigma$. No voter assigns the same value to $\{f\}$ and an outcome of the form $W_\sigma$. If a voter 
 has two actions such that one action results in $W_\sigma$ while the other action results in $W_{\sigma'}$ with 
 $\sigma\neq\sigma'$, then either 
 the voter is not indifferent between
 $W_\sigma$ and $W_{\sigma'}$ (and the difference in values exceeds the cost of voting), 
 or one of the two actions is abstention (and hence has a lower cost).
 However, there may be multiple SPNE: e.g., for a clause voter, it may be the case that abstention is suboptimal, but both voting for her clause candidate and voting for $f$ results in $\{f\}$ being the winner set.
 
 The proof of Theorem~1 goes through even if $\varepsilon=0$, i.e., voting is costless.
\end{remark}
We note that our result does not fully resolve the conjecture of \citet{DE10}, as their conjecture applies to unrestricted ballots (i.e., when every voter can vote for every candidate). We strongly believe that this problem is $\PSPACE$-hard as well;
however, we expect the proof to be less elegant than the one presented here.

\bibliographystyle{plainnat}
\bibliography{ref}
\end{document}